\documentclass[11pt,a4paper]{article}
\usepackage[a4paper,margin=1.1in]{geometry}
\usepackage[T1]{fontenc}
\usepackage[utf8]{inputenc}
\usepackage{lmodern}
\usepackage{microtype}
\usepackage{amsmath,amssymb,amsthm}
\usepackage{enumerate}      
\usepackage{booktabs}
\usepackage{tikz}
\usetikzlibrary{patterns,decorations.pathreplacing}
\usepackage{graphicx}
\usepackage[colorlinks=true,linkcolor=blue!60!black,citecolor=blue!60!black,urlcolor=blue!60!black]{hyperref}
\usepackage[capitalize,noabbrev]{cleveref}

\theoremstyle{plain}
\newtheorem{theorem}{Theorem}
\newtheorem{lemma}{Lemma}
\newtheorem{corollary}{Corollary}
\newtheorem{observation}{Observation}
\crefname{observation}{Observation}{Observations}
\crefname{equation}{}{}   

\newcommand{\OPT}{\operatorname{OPT}}
\newcommand{\problem}{\textnormal{\textsc{Scheduling with Mandatory Breaks}}}
\newcommand{\Ecl}{\mathcal{E}}
\newcommand{\Lcl}{\mathcal{L}}

\newcommand{\aidisclosure}{The authors used generative AI models (Claude Fable 5 and Claude Fable 5.1 by Anthropic in the Claude Code harness,
and GPT 5.6 Sol by OpenAI in the Codex harness) in the preparation of this work.
The models were used to simplify substantially an earlier
version of the hardness proof and to search for simple numerical constants realizing the variable gadget,
for literature search, and for assistance with the writing of the paper and the production of its figures.
All results and proofs were verified by the authors, who take full responsibility for the content.}

\title{Scheduling with Mandatory Breaks:\\ NP-Hardness and an Additive-One Approximation}
\author{Shreyas Ghildiyal\thanks{International Institute of Information Technology Bangalore, India.
  \texttt{shreyas.ghildiyal@iiitb.ac.in}}
  \and
  Muralidhara V.N.\thanks{International Institute of Information Technology Bangalore, India.
  \texttt{murali@iiitb.ac.in}}}
\date{}

\begin{document}
\maketitle

\begin{abstract}
In classical fixed-interval scheduling, each job of a given set must
be processed during a prescribed time interval. The goal is to assign each job to exactly 
one machine such that no two jobs assigned to the same machine overlap in their interiors, and 
the number of machines used is minimized. 
Without further constraints this is interval graph coloring and is solvable in polynomial 
time through greedy approaches. We study a variant in which every used machine must 
remain idle during a contiguous \emph{break} of prescribed length $x$ 
somewhere in the scheduling horizon. We show that this additional constraint makes 
the problem hard, except when $x\le1$. First, for every fixed $x\ge2$, deciding whether $k$ machines suffice
for the assignment of a given set of jobs is NP-complete, even when all coordinates 
are bounded linearly in the number of jobs, so machine minimization is strongly NP-hard. Second, for $x=1$, we give a
polynomial-time algorithm that solves the problem exactly. Finally, we give a
deterministic polynomial-time algorithm that, on every feasible instance, outputs a schedule using
at most $\OPT+1$ machines, where $\OPT$ is the true minimum. Unless 
$\mathrm{P}=\mathrm{NP}$, no polynomial-time algorithm guarantees $\OPT$ machines for any fixed $x\ge2$, so 
the additive guarantee of one is best possible.
\end{abstract}

\section{Introduction}\label{sec:intro}

In classical fixed-interval scheduling, each job $j$ from a given set of jobs must be processed 
during a given time interval $[l_j,r_j]$. The goal is to assign each job to exactly 
one machine (from an unlimited supply of identical machines) such that no two jobs assigned to the same machine conflict, and 
the number of machines used is minimized. Two jobs conflict precisely when their interiors
intersect, so the minimum number of machines equals the chromatic number of the interval graph induced by 
the given set of jobs, which is simply the maximum number of jobs whose interiors contain a common point.
An optimal schedule can be found using various greedy strategies in $O(n\log n)$
time~\cite{GuptaLeeLeung1979}. The model and its many variants have been studied extensively~\cite{KolenLPS2007}.

This paper studies a variant in which every machine that processes at least one job must be completely
idle during some contiguous interval of prescribed length $x$ inside the scheduling horizon. The requirement models a
mandatory rest. In railway crew rostering, for instance, a worker's weekly duties are train journeys that run at fixed
times, and every working crew member is entitled to an uninterrupted weekly rest period. Rest requirements of this kind
are standard throughout crew and personnel rostering~\cite{CapraraTothVigoFischetti1998,ErnstJKS2004,HeilHoffmannBuscher2020}.
Such scenarios map cleanly onto our problem, with train journeys being analogous to jobs, workers to machines, and a week 
to the scheduling horizon. Note that for $x=0$, the break requirement is vacuous and the problem reduces to classical 
interval scheduling.

\paragraph*{Problem definition.}
Formally, a problem instance consists of

\begin{itemize}
  \item an integer \emph{break length} $x$ where $x \ge 0$,
  \item a \emph{scheduling horizon} $[0,y]$, where $y \ge x$ is an integer, and
  \item a set of $n$ \emph{jobs} $J_j = [l_j, r_j]$ with integer endpoints $0 \le l_j < r_j \le y$.
\end{itemize}

A \emph{schedule} assigns every job to a machine so that jobs on the same machine have pairwise disjoint interiors.
Formally, no two jobs $J_a, J_b$ (with $a \neq b$) assigned to the same machine can satisfy $(l_a < r_b) \land (l_b < r_a)$. 
Note that jobs on the same machine can touch at their boundaries (e.g., $l_a < r_a = l_b < r_b$ is an allowed configuration). 
A schedule is said to be \emph{feasible} if and only if for every used machine $m$, there exists an interval 
$[b,b+x]\subseteq[0,y]$ such that for every job $J_j = [l_j, r_j]$ assigned to machine $m$, the condition 
$(r_j \leq b) \lor (b + x \leq l_j)$ holds.

The objective is to minimize the number of used machines in a feasible schedule. For a feasible instance, the minimum is denoted by $\OPT$.
The decision problem \problem{} asks, given additionally an integer $k$, whether $k$ machines suffice.

\paragraph*{Our results.}
A schedule has to balance two competing goals: using few machines, and leaving on every used machine an idle gap
long enough for the break. Our first result shows that for every break length $x\ge2$ this makes the problem
strongly NP-hard.

\begin{theorem}\label{thm:intro-hard}
For every fixed break length $x\ge2$, \problem{} is NP-complete, even when restricted to instances in which all
coordinates are bounded by a linear function of the number of jobs. In particular, minimizing the number of machines is strongly NP-hard.
\end{theorem}

The threshold in \cref{thm:intro-hard} is exact. For $x=0$ the problem is classical interval scheduling, and our
second result shows that the case $x=1$, where any idle gap can host the break because all data are integral, is
polynomial-time solvable as well: the problem becomes hard precisely at $x=2$.

\begin{theorem}\label{thm:intro-x1}
There is an algorithm that, given a break length $x\le1$, a horizon $y$, and a set of $n$ jobs, either reports
that no feasible schedule exists or computes an optimal schedule, in $O(n^2\log n)$ time.
\end{theorem}

Together with \cref{thm:intro-hard}, this settles the complexity of \problem{} for every fixed break length: it is
solvable in polynomial time for $x\le1$ and NP-complete for $x\ge2$.

Our third result shows that the difficulty of \problem{} is nevertheless as small as it can be, because a feasible schedule 
that uses at most one more machine than the optimum can be computed efficiently.

\begin{theorem}\label{thm:intro-approx}
There is a deterministic polynomial-time algorithm that, given a break length $x$, a horizon $y$, and a set of jobs,
either reports that no feasible schedule exists or outputs a feasible schedule using at most $\OPT+1$ machines.
\end{theorem}

A polynomial-time algorithm that always returns an optimal schedule would in particular decide \problem{},
which Theorem~\ref{thm:intro-hard} rules out for every fixed $x\ge2$ unless $\mathrm{P}=\mathrm{NP}$. Since the
number of machines is an integer, the smallest additive guarantee still achievable in polynomial time is one, and
Theorem~\ref{thm:intro-approx} achieves it.

\begin{corollary}\label{cor:best-possible}
Unless $\mathrm{P}=\mathrm{NP}$, the additive guarantee of Theorem~\ref{thm:intro-approx} is the best possible:
for every fixed $x\ge2$, no polynomial-time algorithm outputs a schedule using $\OPT$ machines on every
feasible instance of \problem{}.
\end{corollary}

\subsection{Related work}\label{sec:related}

Fixed-interval scheduling with machine minimization is solvable in polynomial time, and several weighted and
cost-based variants, such as selecting a maximum-value subset of jobs for a given number of machines, are
solvable in polynomial time as well~\cite{ArkinSilverberg1987,KolenLPS2007,KovalyovNgCheng2007}. Adding a per-machine resource constraint changes
this: Fischetti, Martello, and Toth proved NP-hardness when each machine's \emph{spread time} --- the span from its
first start to its last finish --- is bounded~\cite{FischettiMartelloToth1987}, and likewise when its total
\emph{working time} is bounded~\cite{FischettiMartelloToth1989}, later complementing both results with fast
approximation algorithms and worst-case analyses~\cite{FischettiMartelloToth1992}. In the same family,
Osorio-Valenzuela et al.\ minimize the number of machines when each machine has a limited workload
capacity~\cite{OsorioValenzuelaEtAl2019}, and the shift minimization personnel task scheduling problem assigns fixed
tasks to a minimum number of workers who come with predetermined shifts and
skills~\cite{KrishnamoorthyErnstBaatar2012}. Our break constraint is of a different kind, as it demands a contiguous idle
interval whose position is chosen by the schedule, separately on every used machine. The two models have one nontrivial point of coincidence. With integral data,
a machine has an idle gap of length one precisely when its total working time is at most $y-1$, so for $x=1$ our
problem is the working-time problem of Fischetti, Martello, and Toth~\cite{FischettiMartelloToth1989} with
per-machine limit $y-1$, the largest limit that constrains anything. That problem is NP-hard for general limits,
and \cref{thm:intro-x1} shows that this boundary case is polynomial. For $x\ge2$, in contrast, the feasibility
of a machine is not determined by its total workload alone, as a machine may be idle for most of the horizon and
still have no idle gap of length $x$.

Breaks themselves have a long history in personnel scheduling, from implicit models of flexible meal breaks in shift
scheduling~\cite{BechtoldJacobs1990} to the break scheduling problem of Widl and Musliu~\cite{WidlMusliu2014}, in
which breaks are inserted into an already fixed shift plan so that the number of employees on duty tracks a given
staffing curve. In those models the demand is an aggregate headcount over time. In ours it is a set of indivisible
fixed jobs, and the number of machines is the objective. A second adjacent line of work makes machines temporarily
unavailable. In interval scheduling with machine availabilities~\cite{AnapolskaEtAl2024}, each machine's
availability interval is given in the input (in the flexible variant, only the assignment of the given intervals to
machines is free), and the question is whether all jobs can be scheduled. Jaykrishnan and
Levin~\cite{JaykrishnanLevin2024} study machines that must go down after processing a prescribed number of jobs,
including a variant that minimizes the number of machines, but their jobs are movable, so the problem is one of
packing rather than of fixed intervals. In aircraft maintenance routing~\cite{GopalanTalluri1998}, a fleet of
aircraft flies fixed flight legs and every aircraft must regularly reach a maintenance station, but the maintenance
opportunities are tied to given locations and times, and the number of aircraft is given rather than minimized. Our model differs from all of
these in the same way: the jobs are fixed intervals, the number of machines is the objective, and each machine's
idle interval is positioned freely by the schedule, subject only to having the prescribed length. To the best of
our knowledge, this combination has not been studied before.

The sensitivity of interval problems to seemingly small structural perturbations is well documented. A classical
example is that coloring circular-arc graphs, unlike interval graphs, is NP-complete~\cite{GJMP1980}. For general background on NP-completeness we refer to Garey and
Johnson~\cite{GareyJohnson1979}.


\subsection{Technical overview}\label{sec:overview}


\paragraph*{Hardness.}

The reduction is from a restriction of satisfiability in which every variable occurs exactly three times, twice
with one sign and once with the other, and every clause has two or three literals. It is NP-complete by a standard
variable-splitting argument in the spirit of Tovey~\cite{Tovey1984}.

We first show NP-completeness for the case where the break length $x$ is equal to $2$. The reduction then extends
to every fixed $x\ge2$ by scaling all coordinates.
From an arbitrary instance of the aforementioned restriction of satisfiability with $k$ literal occurrences in total, we construct an instance of \problem{} with $x = 2$ and three families of
jobs --- \emph{variable jobs}, \emph{literal jobs}, and \emph{clause jobs} --- with each family consisting of $k$ jobs 
that pairwise conflict within their family. 
These jobs are constructed in a manner such that a schedule on fewer than $k$ machines is impossible, and a schedule
on $k$ machines exists if and only if a solution to the restricted satisfiability instance exists.

\paragraph*{Break length one.}
For $x=1$ the break constraint is volumetric: since all data are integral, a machine has an idle gap of length one
precisely when its jobs do not tile the entire horizon, and the only way to tile the horizon is a sequence of
jobs, each ending where the next one starts, all on one machine. Call it a \emph{handoff} when two touching jobs
are placed consecutively on one machine. A handoff saves a machine at the point where the two jobs touch, and a counting argument shows
that a budget of $K$ machines forces a prescribed number of handoffs at every point, while additional handoffs
never help. What remains is to choose which ending job hands off to which starting job so that no sequence of
handoffs runs from $0$ to $y$.
Following the jobs that start at $0$ through the handoffs, this is a maximum-flow problem, and a binary search over
$K$ finds the optimum.

\paragraph*{The additive-one algorithm.}
The algorithm rests on a normal form: we view the time horizon of every machine in terms of the prefix before its 
break, and the suffix after it. Every job on the machine accordingly receives a \emph{label}: \emph{early} if it lies in the prefix, \emph{late}
if it lies in the suffix. If we knew, for every job, whether it is early or late in some optimal 
schedule, then it is possible to recover an optimal schedule in polynomial time --- one may construct a partial order 
on the jobs such that choosing a chain in this partial order corresponds to choosing a set of jobs that may all be assigned 
to a single machine. Finding the minimum number of machines then reduces to finding the minimum chain cover on this poset.
The standard proof of Dilworth's theorem in turn reduces minimum chain cover to bipartite matching, firmly placing the subproblem in P.
The entire difficulty of the problem therefore lies in choosing the labels. Recording them as $0/1$ values, we show
that the number of machines needed under a given labeling admits a closed formula: a maximum of $O(n^2)$
expressions, each linear in the labels. Minimizing this maximum over fractional labels is a linear program, and its optimum is at most $\OPT$.
Minimizing over \emph{integral} labels is what makes the problem hard. We then exploit the structure of
the constraints: the labels enter them only through their sums over elementary intervals, and the matrix of these
sums has consecutive ones in every column. Such matrices are totally unimodular, and this allows an optimal
fractional solution (which can be found in polynomial time) to be rounded to integral labels while changing every
sum by less than one, and hence every constraint by less than two. Since machine counts
are integers, at most one extra machine is incurred.

\section{Preliminaries}\label{sec:prelim}


Throughout, we assume $x\ge1$, since for $x=0$ the problem is classical interval scheduling, which the greedy
algorithm solves exactly. The number
of jobs in an instance of \problem{} is denoted by $n$. 
A machine is \emph{used} if at least one job is assigned to it.

\begin{observation}\label{obs:np}
\problem{} is in NP.
\end{observation}

\begin{proof}
A schedule is a polynomial-size certificate that can be verified in polynomial time: for each machine, sort its jobs 
by start time and check that consecutive jobs do not conflict and that the longest idle interval --- including the gap before the 
first job and the gap after the last one --- has length at least $x$.
\end{proof}

The following characterization shows that feasibility itself is easy to test, and the hardness of \problem{} lies
entirely in the number of machines.

\begin{observation}\label{obs:feasible}
An instance admits a feasible schedule on some number of machines if and only if every job $j$ satisfies
$r_j\le y-x$ or $l_j\ge x$.
\end{observation}

\begin{proof}
A break has positive length, so every job of a machine lies entirely before or entirely after the break: jobs
before it end by $b\le y-x$, and jobs after it start at $b+x\ge x$. A job violating both conditions therefore fits
on no machine. Conversely, if every job satisfies one of the conditions, placing each job on a private machine is
feasible (a machine whose job ends by $y-x$ takes the break $[y-x,y]$, and any other machine takes $[0,x]$).
\end{proof}

Finally, sort the distinct job endpoints together with $0$ and $y$. We define each open interval between consecutive
values in this sorted order to be an \emph{elementary interval}. There are at most $2n+1$ of them.
The \emph{depth} $d(P)$ of an elementary interval $P$ is the number of jobs $j$ with $P\subseteq(l_j,r_j)$. The maximum 
depth, denoted $D$, is the largest number of jobs whose interiors contain a common point, and also the minimum number of machines when no breaks are
required (classical interval scheduling).

\section{NP-completeness}\label{sec:hardness}


In this section we prove \cref{thm:intro-hard}. Membership in NP is \cref{obs:np}. We establish hardness by a
reduction from the following restriction of satisfiability, which is NP-complete by a standard variable-splitting
argument in the spirit of Tovey~\cite{Tovey1984}.

\begin{lemma}\label{lem:sat}
SAT remains NP-complete for formulas whose clauses contain two or three literal occurrences and in which every variable
occurs exactly three times, twice with one sign and once with the other.
\end{lemma}

\begin{proof}
Take the standard NP-complete restriction of \textsc{3-Sat} in which every clause has exactly three literal
occurrences~\cite{GareyJohnson1979}. For a variable $z$ with occurrences
$1,\dots,t$, replace the $i$th occurrence by a fresh variable $z_i$ of the same sign and add the clauses
$(\neg z_i\lor z_{i+1})$ for $i=1,\dots,t$, indices modulo $t$. These clauses form a cycle of implications, so in
any satisfying assignment all copies of $z$ take the same value, and satisfiability is preserved. Each copy $z_i$
occurs once in an original clause and once with each sign in the cycle, which is the required pattern. For $t=1$
the added clause is the tautological $(\neg z_1\lor z_1)$, which is harmless, and variables without occurrences are
discarded. A clause may then contain both
signs of a variable, which the construction below, being defined per occurrence, tolerates.
\end{proof}

Fix such a formula $\varphi$ with variables $i=0,\dots,p-1$ and clauses $c=0,\dots,q-1$, where clause $c$ has
$s_c\in\{2,3\}$ literals. Counting literal occurrences in two ways gives $\sum_c s_c=3p$.

\subsection{The reduction}\label{sec:construction}

Given $\varphi$ (we may assume $p\ge1$), we construct an instance of \problem{} with $n=9p$ jobs. Set
\[
x=2,\qquad k=3p,\qquad y=6p+3q .
\]

We also define $S_c=6p+3c\quad(0\le c<q)$ as the \emph{station} of clause $c$. The jobs for this instance will come in three families of $3p$ jobs each.

\begin{itemize}
\item\textbf{Variable jobs.} Each variable $i$ contributes the three jobs
\[
[0,\,6i+1],\qquad[0,\,6i+3],\qquad[0,\,6i+4].
\]
\item\textbf{Literal jobs.} Consider variable $i$, and let its two occurrences of the repeated sign lie in
clauses $a_i$ and $b_i$ and its remaining occurrence in clause $e_i$. The variable contributes one job per occurrence,
\[
A_i=[6i+3,\;S_{a_i}],\qquad N_i=[6i+4,\;S_{e_i}],\qquad B_i=[6i+5,\;S_{b_i}],
\]
so that $A_i$ and $B_i$ represent the repeated sign and $N_i$ the opposite sign. Each literal job runs from its
variable's region to the station of its clause.
\item\textbf{Clause jobs.} Each clause $c$ contributes one \emph{filter} $F_c=[S_c+1,\,y]$ and $s_c-1$ identical
\emph{sinks} $[S_c+2,\,y]$.
\end{itemize}

The construction is clearly computable in polynomial time, and all coordinates are at most $y=6p+3q\le10.5\,p$,
since every clause has at least two literals. We claim that $\varphi$ is satisfiable if and only if the
constructed jobs can be feasibly scheduled on $k$ machines. Further, it is impossible to create a feasible schedule with
fewer than $k$ machines.

Before proving this, we outline the intended mechanism. Every machine works exactly one job from each family,
and we will show that the three jobs on a machine are forced to be \emph{compatible}: the literal job belongs to
the variable of the variable job, and the clause job to the clause of the literal job. A machine thus stands for
one literal occurrence of $\varphi$. Every machine is busy at time $0$ and at time $y$, so its only gaps lie
between its variable job and its literal job, or between its literal job and its clause job, and the offsets in
the construction decide which of these gaps are long enough for a break. This is what makes the value assignments
of a variable consistent across its occurrences. \cref{fig:example} illustrates the construction on a concrete
formula.

\begin{figure}[t]
\centering
\begin{tikzpicture}[xscale=0.58,yscale=0.55]
\newcommand{\vjob}[3]{\fill[black!15] (#1,#3) rectangle (#2,#3+0.62);}
\newcommand{\ljob}[4]{\fill[blue!20] (#1,#3) rectangle (#2,#3+0.62);
  \node[font=\scriptsize] at ({(#1+#2)/2},{#3+0.31}) {#4};}
\newcommand{\cjob}[4]{\fill[orange!30] (#1,#3) rectangle (#2,#3+0.62);
  \node[font=\scriptsize] at ({(#1+#2)/2},{#3+0.31}) {#4};}
\newcommand{\brk}[2]{\fill[pattern=north east lines,pattern color=teal] (#1,#2+0.05) rectangle (#1+2,#2+0.57);}
\foreach \s/\l in {12/{$S_0$},15/{$S_1$},18/{$S_2$}}
  {\draw[dashed,black!35] (\s,-0.6) -- (\s,7.4);
   \node[font=\scriptsize] at (\s,7.75) {\l};}
\draw[->] (-0.3,-0.6) -- (21.8,-0.6);
\foreach \t in {0,3,6,9,12,15,18,21}
  {\draw (\t,-0.72) -- (\t,-0.48); \node[below,font=\scriptsize] at (\t,-0.68) {$\t$};}
\vjob{0}{1}{6.6}  \brk{1}{6.6}  \ljob{3}{12}{6.6}{$A_u$}  \cjob{13}{21}{6.6}{$F_0$}
\vjob{0}{3}{5.5}  \brk{3}{5.5}  \ljob{5}{18}{5.5}{$B_u$}  \cjob{19}{21}{5.5}{$F_2$}
\vjob{0}{4}{4.4}  \ljob{4}{15}{4.4}{$N_u$}  \brk{15}{4.4}  \cjob{17}{21}{4.4}{}
\vjob{0}{7}{3.0}  \brk{7}{3.0}  \ljob{10}{15}{3.0}{$N_v$}  \cjob{16}{21}{3.0}{$F_1$}
\vjob{0}{9}{1.9}  \ljob{9}{12}{1.9}{$A_v$}  \brk{12}{1.9}  \cjob{14}{21}{1.9}{}
\vjob{0}{10}{0.8} \ljob{11}{18}{0.8}{$B_v$} \brk{18}{0.8}  \cjob{20}{21}{0.8}{}
\draw[decorate,decoration={brace}] (-0.55,7.3) -- (-0.55,4.35);
\node[font=\small,left] at (-0.75,5.85) {$u$};
\draw[decorate,decoration={brace}] (-0.55,3.7) -- (-0.55,0.75);
\node[font=\small,left] at (-0.75,2.25) {$v$};
\end{tikzpicture}
\caption{The instance produced from $\varphi=(u\lor v)\land(\neg u\lor\neg v)\land(u\lor v)$, which satisfies
the restriction of \cref{lem:sat}, with variables $u,v$ and clauses $0,1,2$ in this order. Here $p=2$, $q=3$, $k=6$,
$y=21$, and the stations are $S_0=12$, $S_1=15$, $S_2=18$ (dashed). Rows are the six machines of a feasible
schedule realizing the satisfying assignment $u=\mathrm{true}$, $v=\mathrm{false}$. Gray jobs are variable jobs,
blue jobs are literal jobs (subscripted by their variable), and orange jobs are clause jobs. The filters are
labeled $F_c$, and the unlabeled clause jobs are the sinks. Hatched teal boxes are the breaks ($x=2$).}
\label{fig:example}
\end{figure}
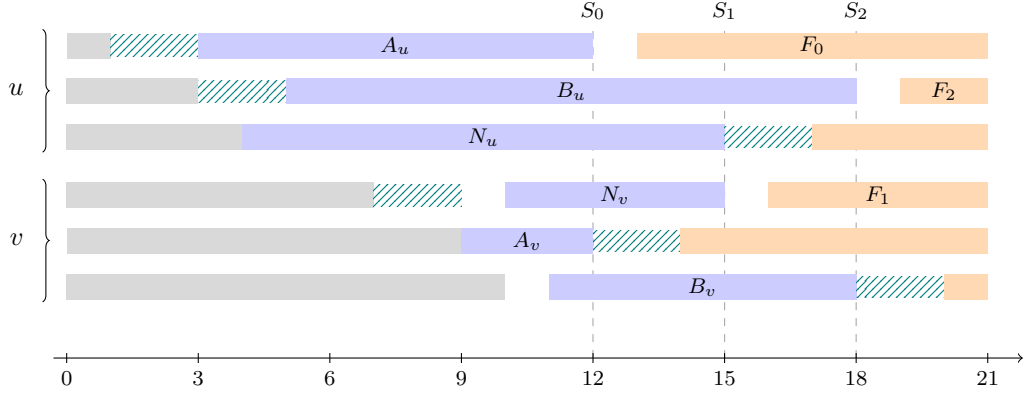

\subsection{Correctness}\label{sec:hardness-proof}


We begin with the structural claims of the outline above.

\begin{lemma}[Structure]\label{lem:structure}
In every schedule of the constructed instance on at most $k$ machines, exactly $k$ machines are used, and each
machine works exactly three jobs: one variable job, one literal job, and one clause job, in this order. Moreover:
\begin{enumerate}[(i)]
\item the three literal jobs of variable $i$ are worked by the three machines that work the variable jobs of
variable $i$;
\item the $s_c$ clause jobs of clause $c$ are worked by the machines that work the $s_c$ literal jobs of the
literals of clause $c$.
\end{enumerate}
\end{lemma}

\begin{proof}
Every variable job contains the interval $[0,1]$. Every literal job contains $[6p-1,6p]$: it starts at
$6(p-1)+5=6p-1$ or earlier and ends at a station, i.e., at $6p$ or later. Every clause job contains $[y-1,y]$.
Each family is therefore a set of $3p=k$ pairwise conflicting jobs and must occupy $k$ distinct machines. Hence
exactly $k$ machines are used, and each works precisely one job from each family --- three jobs in total, and
non-conflict forces the order variable job, literal job, clause job.

Next we prove (i). Since each machine works exactly one job from each family, it suffices to show that, for
every $i$, the set of machines working the literal jobs of variables $0,\dots,i$ equals the set of machines
working the variable jobs of variables $0,\dots,i$. Statement (i) then follows by taking differences for
consecutive values of $i$. A variable job and a literal job on the same machine must not conflict, and variable
jobs start at time $0$, so a machine's variable job must end no later than its literal job starts. Now consider
any machine working a literal job of some variable $j\le i$. That literal job starts no later than $6j+5\le6i+5$, so the
machine's variable job ends at $6i+5$ or earlier; since the variable jobs of variables $i+1,\dots,p-1$ end at
$6(i+1)+1=6i+7$ or later, the machine's variable job belongs to one of the variables $0,\dots,i$. The $3(i+1)$
literal jobs of variables $0,\dots,i$ thus occupy $3(i+1)$ distinct machines, each of whose variable jobs
belongs to variables $0,\dots,i$ --- and there are exactly $3(i+1)$ machines with the latter property, one per
such variable job. The two sets of machines therefore coincide.

The proof of (ii) is analogous: a machine's literal job must end no later than
its clause job starts; the clause jobs of clauses $0,\dots,c$ start at $S_c+2$ or earlier, while the literal
jobs of clauses $c+1,\dots,q-1$ (if any) end at $S_{c+1}=S_c+3$ or later; and the number of clause jobs of clauses
$0,\dots,c$ and the number of literal jobs of those clauses both equal $\sum_{c'\le c}s_{c'}$.
\end{proof}

By \cref{lem:structure}, each machine works a variable job, a literal job of the same variable, and a clause job
of its literal's clause, so machines correspond to the literal occurrences of $\varphi$.
Since every machine is busy at time $0$ and at time $y$, its idle time consists of exactly two (possibly empty) gaps: a
gap $g_1$ between its variable job and its literal job, and a gap $g_2$ between its literal job and its clause
job. A schedule on $k$ machines is therefore feasible if and only if every machine has $g_1\ge2$ or $g_2\ge2$.
Call a literal job \emph{safe} (in a given schedule) if the machine working it has $g_1\ge2$.

Consider variable $i$ and any schedule on at most $k$ machines. Denote by $M_1$, $M_3$, and $M_4$ the machines working the three variable jobs of
variable $i$, indexed by those jobs' right endpoints relative to $6i$. By \cref{lem:structure}(i), the machines $M_1,M_3,M_4$ work exactly
the literal jobs $A_i,N_i,B_i$, one each. Call the resulting assignment of literal jobs to these machines the \emph{routing} of the gadget.
A routing is \emph{conflict-free} if each of the three machines' variable job ends no later than its literal job starts.

\begin{lemma}[Variable gadget]\label{lem:switch}
In every schedule of the constructed instance on at most $k$ machines, the set of safe literal jobs of variable $i$ is one of
\[
\{A_i\},\qquad\{B_i\},\qquad\{A_i,B_i\},\qquad\{N_i\};
\]
that is, a nonempty set of occurrences of a single sign. Conversely, each of these four sets is the safe set of a conflict-free
routing of the gadget.
\end{lemma}

\begin{proof}
Work in coordinates relative to $6i$, so that the variable jobs of
variable $i$ end at $1$, $3$, and $4$, and its literal jobs $A_i,N_i,B_i$ start at $3$, $4$, and $5$.
Since variable jobs start at time $0$, the only pairing that is not conflict-free is $M_4$ working $A_i$. Every other pairing
is legal, and its first gap $g_1$ equals the start of the literal job minus the end of the variable job, as
listed in \cref{tab:gadget}. A pairing is safe exactly if this difference is at least $2$ (the four
entries $(1,A_i)$, $(1,N_i)$, $(1,B_i)$, and $(3,B_i)$, which are bold in \cref{tab:gadget}).

Of the six bijective assignments of $A_i,N_i,B_i$ to $M_1,M_3,M_4$, exactly two pair $A_i$ with $M_4$, so there
are exactly four conflict-free routings. By \cref{tab:gadget}, their safe pairings are as follows:
\begin{itemize}
\item $M_1{\to}A_i$, $M_3{\to}N_i$, $M_4{\to}B_i$: only $(1,A_i)$ is safe, and the safe set is $\{A_i\}$;
\item $M_1{\to}A_i$, $M_3{\to}B_i$, $M_4{\to}N_i$: $(1,A_i)$ and $(3,B_i)$ are safe, and the safe set is
$\{A_i,B_i\}$;
\item $M_1{\to}N_i$, $M_3{\to}A_i$, $M_4{\to}B_i$: only $(1,N_i)$ is safe, and the safe set is $\{N_i\}$;
\item $M_1{\to}B_i$, $M_3{\to}A_i$, $M_4{\to}N_i$: only $(1,B_i)$ is safe, and the safe set is $\{B_i\}$.
\end{itemize}
The set of safe literal jobs of variable $i$ is therefore always one of the four claimed sets.

Conversely, each of the four sets is the safe set of one of the four routings above, and these routings are
conflict-free.
\end{proof}

\begin{table}[t]
\caption{The first gap $g_1$ for each pairing in the gadget of variable $i$, in coordinates relative to $6i$.
Bold entries ($g_1\ge2$) are the safe pairings. The variable job ending at $4$ cannot be paired with $A_i$, as
the two jobs conflict.}
\label{tab:gadget}
\centering
\begin{tabular}{@{}lccc@{}}
\toprule
& $A_i$ (starts $3$) & $N_i$ (starts $4$) & $B_i$ (starts $5$)\\
\midrule
variable job ending $1$ & $\mathbf{2}$ & $\mathbf{3}$ & $\mathbf{4}$\\
variable job ending $3$ & $0$ & $1$ & $\mathbf{2}$\\
variable job ending $4$ & --- & $0$ & $1$\\
\bottomrule
\end{tabular}
\end{table}

We can now prove \cref{thm:intro-hard}.

\begin{proof}[Proof of \cref{thm:intro-hard}]
Membership in NP is \cref{obs:np}. For hardness, consider first $x=2$. Given a formula $\varphi$ as in
\cref{lem:sat}, the construction of \cref{sec:construction} produces an instance of \problem{} together with a
target $k$, in polynomial time. We now prove that this instance admits a feasible schedule on at most $k$
machines if and only if $\varphi$ is satisfiable. The construction is then a polynomial-time reduction from an
NP-complete problem, and hardness follows.

By \cref{lem:structure}(ii), the
clause jobs of clause $c$ are worked by the machines of its own $s_c$ literal jobs. Every literal job of clause
$c$ ends at the station $S_c$, the filter $F_c$ starts at $S_c+1$, and the sinks start at $S_c+2$. Hence, the
filter gives its machine $g_2=1$ while each sink gives $g_2=2$.

Suppose first that $\varphi$ has a satisfying assignment. We construct a feasible schedule on $k$ machines.
Consider a variable $i$. The assignment makes the literals of exactly one of its two signs true. Assign the three
variable jobs and the three literal jobs of variable $i$ to three machines according to the routing of
\cref{lem:switch} whose safe set is $\{A_i,B_i\}$ if the true sign is the repeated one, and $\{N_i\}$ otherwise.
This routing is conflict-free, and it makes a literal job safe if and only if its literal is true under the
assignment. Next, consider a clause $c$. Its $s_c$ clause jobs must be
worked by the machines that work the literal jobs of its $s_c$ literals, and any assignment between the two is
conflict-free, since every literal job of clause $c$ ends at $S_c$ and every clause job of clause $c$ starts at
$S_c+1$ or later. As the assignment satisfies $\varphi$, clause $c$ contains at least one true literal. Choose
one, let the machine working its literal job work the filter $F_c$, and distribute the sinks among the remaining
$s_c-1$ machines arbitrarily. In the resulting schedule, every machine working a sink has $g_2=2$, and every
machine working a filter has a safe literal job and hence $g_1\ge2$. Every machine therefore has an idle gap of
length at least $2$ in which to place its break, and the schedule is feasible on $k$ machines.

Conversely, suppose that a feasible schedule on at most $k$ machines exists. We show how to construct a satisfying
assignment for $\varphi$. By \cref{lem:switch}, the safe literal jobs of each variable form a nonempty set of
occurrences of a single sign. Assign to each variable the value that makes the literals of this sign true. Now
consider an arbitrary clause $c$ and the machine working its filter $F_c$. This machine has $g_2=1$, so, since
the schedule is feasible, it has $g_1\ge2$, and its literal job is therefore safe. By \cref{lem:structure}(ii),
this literal job belongs to a literal of clause $c$, and since it is safe, this literal is true under the
constructed assignment. Hence every clause of $\varphi$ contains a true literal, and $\varphi$ is satisfiable.

Moreover, the constructed instance has $n=9p$ jobs with all coordinates at most $10.5\,p$, so \problem{} is
NP-complete for $x=2$ even in unary encoding, and machine minimization is strongly NP-hard.
Finally, to show NP-hardness for an arbitrary fixed break length $x\ge2$, keep $k$, set the break length to $x$, and multiply the horizon and every job endpoint of the construction by $x-1$.
Idle gaps of length at most $1$ become at most $x-1<x$, and idle gaps of length at least $2$ become at least
$2x-2\ge x$. With \emph{safe} now meaning $g_1\ge x$, the counting arguments are invariant under scaling and the
proof goes through unchanged. All coordinates remain bounded by $10.5\,(x-1)\,p$, which is linear in $n$ for fixed $x$.
\end{proof}

\section{A polynomial-time algorithm for break length one}\label{sec:x1}


In this section we prove \cref{thm:intro-x1}. Throughout, $x=1$. By \cref{obs:feasible}, the instance is infeasible precisely when some job equals $[0,y]$. The
algorithm rejects this case first, and we assume from now on that no job equals $[0,y]$. Everything rests on the
following observation.

\begin{observation}\label{obs:volume}
For $x=1$, a machine has an idle interval of length $1$ inside $[0,y]$ if and only if the total length of its
jobs is less than $y$, that is, if and only if its jobs do not cover all of $[0,y]$.
\end{observation}

\begin{proof}
The jobs of a machine have pairwise disjoint interiors inside $[0,y]$, so the idle time of the machine consists of
finitely many gaps of integral length whose total length is $y$ minus the total length of the jobs. Some gap has
length at least $1$ if and only if this total is positive.
\end{proof}

Jobs on one machine cover $[0,y]$ only if, at the right endpoint of every one of them, either the horizon ends or
another job begins. The decisions that matter are therefore which touching jobs share a machine. For an integer $t$ with $0\le t\le y$, let $E_t$ and $S_t$ be
the sets of jobs ending and starting at $t$, and let $\delta(t)$ be the number of jobs $j$ with $l_j\le t\le r_j$.
Note that $E_0=S_y=\emptyset$.

A \emph{handoff set} is a set $H$ of arcs $i\to j$ between jobs with $r_i=l_j$ in which every job has at most one
incoming and at most one outgoing arc. An arc $i\to j$ with $r_i=l_j=t$ is a \emph{handoff at $t$}, and $h_t$
denotes the number of handoffs at $t$. Since $r_j>r_i$ for every arc $i\to j$, the digraph with vertex set the jobs and arc set $H$ is a union of
vertex-disjoint directed paths, possibly consisting of a single job, which we call \emph{runs}. The jobs of a run tile an interval, the \emph{span} of the run, and a
run \emph{covers the horizon} if its span is $[0,y]$.

\begin{lemma}[Handoffs and machines]\label{lem:handoffs}
The jobs can be feasibly scheduled on at most $K$ machines if and only if there is a handoff set in which no run
covers the horizon and $\delta(t)-h_t\le K$ for every integer $t\in[0,y]$.
\end{lemma}

\begin{proof}
Given a feasible schedule on at most $K$ machines, let $H$ contain the arc $i\to j$ whenever $i$ and $j$ are on the same
machine and $r_i=l_j$. Every job has at most one such predecessor and at most one such successor, so $H$ is a
handoff set, and its runs are the maximal sequences of touching jobs on the machines. At an integer $t$, a machine
carries at most two jobs whose closed intervals contain $t$, and if it carries two, then one of them ends and the
other starts at $t$, so they form a handoff at $t$. Hence every machine contributes at most one to $\delta(t)-h_t$,
and $\delta(t)-h_t\le K$. A run covering the horizon would be a machine whose jobs cover $[0,y]$, contradicting
\cref{obs:volume}.

Conversely, let $H$ be a handoff set as in the statement. We assign runs to machines so that the spans of the runs
on one machine are pairwise disjoint as closed intervals. Unlike jobs, two runs whose spans merely touch at a
point are thus kept on different machines, since placing them on one machine would amount to an additional
handoff. This is still interval graph coloring (since spans have integer endpoints, extending every span by $\tfrac12$
at its right end turns it into the usual convention in which only interiors must be disjoint), so the number of machines needed is the
largest number of spans containing a common point, and the greedy left-to-right sweep achieves this
bound~\cite{GuptaLeeLeung1979}. The span of a run contains an integer $t$ if and only if the run contains a job
whose closed interval contains $t$, and a run contains two such jobs exactly when it has a handoff at $t$. So
exactly $\delta(t)-h_t\le K$ spans contain $t$. A non-integer point is contained in no more spans than the next
integer is, because spans have integer endpoints. Thus $K$ machines suffice. Consecutive jobs of a run touch, and
the spans of the runs on one machine are disjoint, so the jobs on every machine have pairwise disjoint interiors.
Finally, a machine carrying two or more runs is idle between them, and a machine carrying a single run has a span
other than $[0,y]$. So the jobs of no machine cover $[0,y]$, and the schedule is feasible by \cref{obs:volume}.
\end{proof}

A handoff lowers $\delta(t)-h_t$ at one point and otherwise only lengthens a run. The next lemma shows that, for
a budget $K$, it is optimal to hand off exactly as often as the budget forces.

\begin{lemma}[Fewest handoffs]\label{lem:fewest}
If some handoff set satisfies the conditions of \cref{lem:handoffs} for $K$, then so does one with
$h_t=\max(0,\,\delta(t)-K)$ for every integer $t\in[0,y]$.
\end{lemma}

\begin{proof}
Among the handoff sets satisfying the conditions, choose one, $H$, with the fewest arcs. The condition
$\delta(t)-h_t\le K$ gives $h_t\ge\delta(t)-K$ at every $t$. Suppose $h_t>\max(0,\delta(t)-K)$ at some $t$, and
remove one handoff at $t$ from $H$. This changes $\delta(t')-h_{t'}$ only at $t'=t$, where the new value
$\delta(t)-h_t+1$ is at most $K$, and it splits one run into two runs with shorter spans. The result satisfies the
conditions with fewer arcs, a contradiction.
\end{proof}

Fix a budget $K\ge D$, where $D$ is the maximum depth of \cref{sec:prelim}, the minimum number of machines
without breaks, and write $c_t=\max(0,\,\delta(t)-K)$. Then $c_t\le\min(|E_t|,|S_t|)$ for every $t$. At $t=0$
this holds because every job of $S_0$ contains the first elementary interval, so $\delta(0)=|S_0|\le D\le K$ and
$c_0=0$, and symmetrically $c_y=0$. At an integer
$0<t<y$, the $\delta(t)-|S_t|$ jobs whose closed intervals contain $t$ but that do not start at $t$ all contain in
their interiors the points immediately to the left of $t$, which lie in a common elementary interval, so
$\delta(t)-|S_t|\le D\le K$, and symmetrically $\delta(t)-|E_t|\le K$.

After \cref{lem:fewest}, what remains is to choose \emph{which} job of $E_t$ hands off to which job of $S_t$, in
such a way that no run covers the horizon. Only a run whose span starts at $0$ can cover the horizon, and such a
run starts with a job of $S_0$. Whenever a job of such a run ends, at a point $t$ say, that job is either one of
the $|E_t|-c_t$ jobs of $E_t$ without an outgoing handoff, and the run ends, or one of the $c_t$ jobs of $E_t$
that hand off, and the run continues on a job of $S_t$. The run has to end at some $t<y$. This is captured by the following network $N_K$. It has a source, a sink,
and two nodes $u_t,v_t$ for every right endpoint $t<y$ of a job, and the following arcs:
\begin{itemize}
\item $\mathrm{source}\to u_{r_j}$ of capacity $1$ for every job $j\in S_0$;
\item $u_t\to\mathrm{sink}$ of capacity $|E_t|-c_t$ and $u_t\to v_t$ of capacity $c_t$, for every right endpoint
$t<y$;
\item $v_t\to u_{r_j}$ of capacity $1$ for every job $j\in S_t$ with $r_j<y$.
\end{itemize}
Parallel arcs are allowed. We call the arc introduced for a job $j$ the \emph{arc of $j$}. Every job has at most
one arc, and a job ending at $y$ has none, since a run starting at $0$ must not continue on it. At an integer $t$
at which no job ends, $c_t\le|E_t|=0$ by the bound above, so no nodes are needed there, and $N_K$ has $O(n)$ nodes
and arcs.

\begin{lemma}[Routing]\label{lem:routing}
Let $K\ge D$. A handoff set with $h_t=c_t$ for every integer $t\in[0,y]$ and no run covering the horizon exists if
and only if $N_K$ admits a flow of value $|S_0|$.
\end{lemma}

\begin{proof}
Let $H$ be such a handoff set. Every job of $S_0$ starts a run $j_1\to j_2\to\dots\to j_m$, and $r_{j_m}<y$ since
the run does not cover the horizon. Route one unit of flow along
$\mathrm{source}\to u_{r_{j_1}}\to v_{r_{j_1}}\to u_{r_{j_2}}\to\dots\to u_{r_{j_m}}\to\mathrm{sink}$. Distinct
runs use distinct jobs, so the arcs of capacity $1$ are respected. At a node $u_t$, the units that continue to
$v_t$ correspond to distinct handoffs at $t$, of which there are $c_t$, and the units that enter the sink
correspond to distinct jobs of $E_t$ without an outgoing handoff, of which there are $|E_t|-c_t$. So the flow is
feasible, and its value is $|S_0|$.

Conversely, let $N_K$ admit a flow of value $|S_0|$. Take it integral and, since $N_K$ is acyclic, decompose it
into $|S_0|$ paths carrying one unit each, one through every source arc. A path enters a node $u_t$ on the arc of
a job of $E_t$ and, if it continues, leaves $v_t$ on the arc of a job of $S_t$. We build $H$ at every $t$ in two steps. First,
for every path that continues through $v_t$, add the arc from the job on which it enters $u_t$ to the job on which
it leaves $v_t$. The arcs of capacity $1$ make these jobs distinct. Second, add arcs from $E_t$ to $S_t$ between
jobs not used so far at $t$, excluding the jobs on which paths enter $u_t$ and stop, until there are $c_t$ arcs at
$t$. This is possible: if $a$ paths enter $u_t$ and $f$ of them continue, then $a-f\le|E_t|-c_t$ by the capacity
of the sink arc, so at least $|E_t|-a\ge c_t-f$ jobs of $E_t$ are available, and at least $|S_t|-f\ge c_t-f$ jobs
of $S_t$ are available since $c_t\le|S_t|$. Every job has at most one incoming and at most one outgoing arc, so
$H$ is a handoff set with $h_t=c_t$ at every $t$. The run starting at a job $j\in S_0$ follows the path through
$\mathrm{source}\to u_{r_j}$: at every node, the job on which the path arrives has an outgoing handoff exactly when
the path continues, and then to the job on which the path leaves. Every path enters the sink from a node $u_t$
with $t<y$, so no run starting at $0$ covers the horizon, and no other run can.
\end{proof}

\begin{proof}[Proof of \cref{thm:intro-x1}]
For $x=0$ the greedy algorithm computes an optimal schedule in $O(n\log n)$ time~\cite{GuptaLeeLeung1979}, so let
$x=1$. We claim that
\[
\OPT=\min\{\,K\ge D:\ N_K\text{ admits a flow of value }|S_0|\,\}.
\]
If $K$ machines suffice, then $K\ge D$, and \cref{lem:handoffs,lem:fewest,lem:routing} in turn yield a handoff set
with the properties of \cref{lem:handoffs}, one with $h_t=c_t$ at every $t$, and a flow of value $|S_0|$.
Conversely, for $K\ge D$ a flow of value $|S_0|$ yields by \cref{lem:routing} a handoff set with $h_t=c_t$ and no
run covering the horizon, and since $\delta(t)-c_t=\min(\delta(t),K)\le K$, \cref{lem:handoffs} turns it into a
feasible schedule on $K$ machines. So the displayed set consists of the budgets $K$ for which $K$ machines
suffice. This set is upward closed, and it contains $K=\max_t\delta(t)$, where every $c_t$ is zero and every job
of $S_0$ simply ends its run at its right endpoint. Hence binary search over $K\in[D,\max_t\delta(t)]$ finds
$\OPT$ with $O(\log n)$ maximum-flow computations.

For the running time, the sets $E_t$ and $S_t$, the values $\delta(t)$ at the endpoints of jobs, and $D$ are
obtained in one sweep in $O(n\log n)$ time. The network $N_K$ has $O(n)$ nodes and arcs, and its maximum flow
value is at most $|S_0|\le n$, so an augmenting-path algorithm, which needs at most $n$ augmentations of $O(n)$
time each, computes it in $O(n^2)$ time. Finally, a schedule
is assembled from the flow for $K=\OPT$ in $O(n\log n)$ time: decompose the flow into paths, which use $O(n)$
arcs in total since distinct paths use distinct arcs of jobs, build the handoff set of \cref{lem:routing}, and
assign runs to machines by the greedy sweep of \cref{lem:handoffs}.
\end{proof}

\section{An additive-one approximation algorithm}\label{sec:algorithm}


In this section we prove \cref{thm:intro-approx}. By \cref{obs:feasible}, which the algorithm checks first, every
job satisfies $r_j\le y-x$ or $l_j\ge x$. Define the two \emph{eligibility classes}
\[
\Ecl=\{\,j: r_j\le y-x\,\},\qquad \Lcl=\{\,j: l_j\ge x\,\},
\]
so that $\Ecl\cup\Lcl$ contains all jobs. A job in $\Ecl$ can be processed before a break, a job in $\Lcl$ after
one, and a job in $\Ecl\cap\Lcl$ either way.

\subsection{A normal form and a chain cover}\label{sec:normal-form}

The break splits the timeline of each machine into a part before the break and a part after it, and every job on
the machine lies entirely in one of the two parts. This gives feasible machines the following normal form.

\begin{lemma}[Normal form]\label{lem:normal-form}
A set of pairwise non-conflicting jobs, listed chronologically, can be feasibly scheduled on one machine if and
only if it is the concatenation of a (possibly empty) prefix of jobs from $\Ecl$ and a (possibly empty) suffix of
jobs from $\Lcl$ such that, if both parts are nonempty, the last prefix job $i$ and the first suffix job $j$
satisfy $l_j-r_i\ge x$.
\end{lemma}

\begin{proof}
Fix a break $[b,b+x]$ of the machine. Since $x\ge1$, every job lies entirely before or entirely after it. Jobs
before it end by $b\le y-x$ and belong to $\Ecl$. Jobs after it start at $b+x\ge x$ and belong to $\Lcl$. If
both groups are nonempty, the gap between them contains the break and has length at least $x$. Conversely, a set
of jobs from $\Ecl$ is feasible with the break $[y-x,y]$, a set of jobs from $\Lcl$ with the break $[0,x]$, and a
mixed set as in the statement accommodates its break in the gap between the two parts.
\end{proof}

Guided by the normal form, we assign each job $j$ a \emph{label} $z_j\in\{0,1\}$, where $z_j=1$ means that $j$ is
processed before its machine's break (\emph{early}) and $z_j=0$ that it is processed after (\emph{late}). A
labeling $z\in\{0,1\}^n$ is \emph{legal} if $z_j=1$ for every $j\in\Ecl\setminus\Lcl$ and $z_j=0$ for every
$j\in\Lcl\setminus\Ecl$. We call these coordinates \emph{forced} and the coordinates of jobs in $\Ecl\cap\Lcl$
\emph{free}. For a legal labeling $z$, let $i\prec_z j$ if $r_i\le l_j$ and either
\begin{enumerate}[(a)]
\item $z_i=z_j$, or
\item $z_i=1$, $z_j=0$, and $l_j-r_i\ge x$.
\end{enumerate}
The relation $\prec_z$ is irreflexive, and it is transitive: a chain $i\prec_z j\prec_z h$ contains at most one
early-to-late transition, and the gap between $i$ and $h$ contains the gap at that transition. Hence $\prec_z$ is
a strict partial order, and \cref{lem:normal-form} translates into the following statement.

\begin{lemma}\label{lem:chains}
Let $z$ be a legal labeling. A set of jobs can be feasibly scheduled on one machine with every job on the side of
the break prescribed by $z$ if and only if it is a chain of $\prec_z$. Consequently, the minimum number of
machines of a feasible schedule respecting $z$ equals the width $W(z)$ of the partial order
$(\{1,\dots,n\},\prec_z)$, and a schedule attaining it can be computed in polynomial time.
\end{lemma}

\begin{proof}
The first claim restates \cref{lem:normal-form}: consecutive jobs of a chain with the same label need only be
non-conflicting, and the single early-to-late transition of a chain carries a gap of length at least $x$. The
minimum number of chains covering a partial order equals its width by Dilworth's theorem~\cite{Dilworth1950}, and
a minimum chain cover is computable in polynomial time by the standard reduction to bipartite matching.
\end{proof}

Conversely, every feasible schedule induces a legal labeling: set $z_j=1$ exactly for the jobs processed before
their machine's break. Legality follows from the proof of \cref{lem:normal-form}, and every machine of the
schedule is a chain of $\prec_z$. Therefore
\begin{equation}\label{eq:opt-minz}
\OPT=\min\{\,W(z): z\in\{0,1\}^n\text{ legal}\,\}.
\end{equation}
The whole difficulty of the problem lies in this minimization over $z$, which, for every fixed $x\ge2$, cannot be
carried out exactly in polynomial time by \cref{thm:intro-hard} unless $\mathrm{P}=\mathrm{NP}$. The remaining two subsections show that it can be solved to within an additive one.

\subsection{A formula for the width}\label{sec:width}

Recall the elementary intervals of \cref{sec:prelim}. Order them from left to right, and let $\mathbf{A}=(a_{Pj})$ be their incidence
matrix with jobs: $a_{Pj}=1$ if $P\subseteq(l_j,r_j)$ and $a_{Pj}=0$ otherwise, so that the depth of $P$ is $d(P)=\sum_j a_{Pj}$. For a
labeling $z$, define the \emph{early depth} of $P$ as
\[
Z_z(P)=\sum_j a_{Pj}\,z_j ;
\]
when $z$ is integral, $Z_z(P)$ counts the early jobs containing $P$ and
$d(P)-Z_z(P)$ the late ones. We call an ordered pair $(P,Q)$ of elementary intervals
\emph{admissible} if the left endpoint of $Q$ minus the right endpoint of $P$ is less than $x$ (in particular, whenever $Q$ does not lie to the right of $P$). Equivalently, $(P,Q)$ is admissible if there are points $s\in P$ and $t\in Q$ with $t-s<x$. The width of
$\prec_z$ is then determined by the depths alone.

\begin{lemma}[Width formula]\label{lem:width}
For every legal labeling $z\in\{0,1\}^n$,
\begin{equation}\label{eq:width}
W(z)=\max_{(P,Q)\ \mathrm{admissible}}\ \bigl(Z_z(P)+d(Q)-Z_z(Q)\bigr).
\end{equation}
\end{lemma}

\begin{figure}[t]
\centering
\begin{tikzpicture}[x=0.9cm,y=0.42cm]
  \draw[->] (0,0) -- (13.2,0);
  \fill[gray!25] (2.5,0) rectangle (4.0,6.6);
  \fill[gray!25] (6.2,0) rectangle (8.2,6.6);
  \node[below] at (3.25,0) {$P$};
  \node[below] at (7.2,0) {$Q$};
  \foreach \l/\r/\h in {0.6/4.0/1, 1.4/5.0/2, 2.5/4.6/3} {
    \draw[thick] (\l,\h) -- (\r,\h);
    \draw[thick] (\l,\h-0.2) -- (\l,\h+0.2);
    \draw[thick] (\r,\h-0.2) -- (\r,\h+0.2);
  }
  \foreach \l/\r/\h in {6.2/9.5/4, 5.4/8.2/5, 5.9/10.8/6} {
    \draw[thick,dashed] (\l,\h) -- (\r,\h);
    \draw[thick] (\l,\h-0.2) -- (\l,\h+0.2);
    \draw[thick] (\r,\h-0.2) -- (\r,\h+0.2);
  }
  \draw[decorate,decoration={brace,amplitude=3pt,mirror}] (4.0,-1.4) -- (6.2,-1.4)
    node[midway,below=3pt] {$<x$};
  \node[right] at (10.9,1) {early jobs};
  \node[right] at (10.9,6) {late jobs};
\end{tikzpicture}
\caption{An antichain of $\prec_z$ as counted by an admissible pair $(P,Q)$: the early jobs (solid) pairwise
conflict and contain $P$, the late jobs (dashed) pairwise conflict and contain $Q$, and the gap between $P$ and $Q$
is shorter than $x$, so no early job precedes a late one (schematic).}
\label{fig:width}
\end{figure}
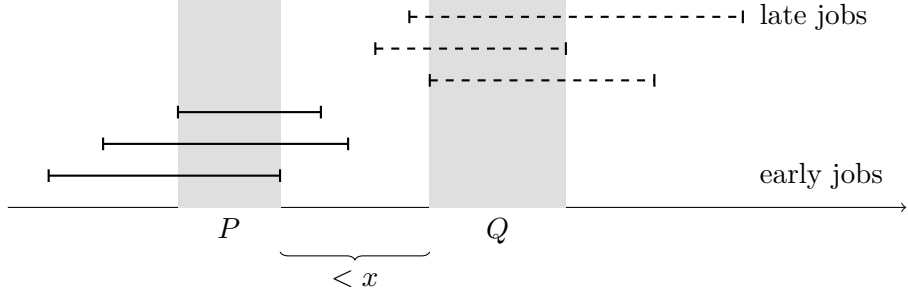

\begin{proof}
We first show that every antichain is counted by some admissible pair (\cref{fig:width}). Let $C$ be a nonempty
antichain (for $C=\emptyset$ there is nothing to show). Two jobs with the same label are comparable as soon as their interiors are disjoint, so the early members of $C$
pairwise conflict, and by the Helly property of intervals their interiors have a nonempty common intersection.
The same holds for the late members. If one of the two groups is empty, choose an elementary interval $P$ inside
the common intersection of the other group and set $Q=P$. The pair is admissible because the left endpoint of $P$ precedes its right endpoint, and
$Z_z(P)+d(P)-Z_z(P)=d(P)\ge|C|$. If both groups are nonempty, let $\rho$ be the smallest right endpoint among
the early members of $C$ and $\lambda$ the largest left endpoint among the late members. Let $P$ be the
elementary interval with right endpoint $\rho$ and $Q$ the elementary interval with left endpoint $\lambda$ --- then
$P$ lies in the common intersection of the early members, since each of them conflicts with the job ending at $\rho$
and hence starts before $\rho$, and likewise $Q$ lies in that of the late members. If
$\rho\le\lambda$, the early job ending at $\rho$ and the late job starting at $\lambda$ have disjoint interiors, so their
incomparability forces $\lambda-\rho<x$; if $\lambda<\rho$, then $\lambda-\rho<0<x$. Thus $(P,Q)$ is
admissible, every early member of $C$ is counted by $Z_z(P)$, and every late member by $d(Q)-Z_z(Q)$, so the
right-hand side of~\eqref{eq:width} is at least $|C|$.

For the converse, fix an admissible pair $(P,Q)$ and points $s\in P$, $t\in Q$ with $t-s<x$. Let $C$ consist of
all early jobs whose interior contains $s$ and all late jobs whose interior contains $t$. No job is counted twice,
since each job is either early or late, so $|C|=Z_z(P)+d(Q)-Z_z(Q)$. Early members of $C$ pairwise conflict and
are therefore incomparable. The same is true for late members. A late job never precedes an early job in $\prec_z$, 
and if an early member $i$ preceded a late member $j$, then $l_j-r_i<t-s<x$, contradicting condition~(b). So $C$ is an
antichain, and $W(z)\ge|C|$.
\end{proof}

\subsection{Relaxing and rounding}\label{sec:rounding}

By~\eqref{eq:opt-minz} and \cref{lem:width}, $\OPT$ is the optimal value of the integer program
\begin{equation}\label{eq:ip}
\begin{array}{ll@{\qquad}l}
\text{minimize} & K \\
\text{subject to} & Z_z(P)+d(Q)-Z_z(Q)\le K & \text{for all admissible pairs }(P,Q),\\
& z_j\in\{0,1\} & \text{for all free }j,
\end{array}
\end{equation}
where the forced coordinates of $z$ are fixed to their legal values. The constraints are linear in $z$, and there
are $O(n)$ variables and $O(n^2)$ constraints. The algorithm relaxes $z_j\in\{0,1\}$ to $0\le z_j\le1$ and solves
the resulting linear program in polynomial time. Let $(z^*,K^*)$ be an optimal solution. Every legal labeling is
feasible for the relaxation, so
\[
K^*\le\OPT.
\]

It remains to round $z^*$ to a legal labeling without losing more than one machine. Recall the incidence matrix $\mathbf{A}$. The
elementary intervals contained in a job are consecutive, so every column of $\mathbf{A}$ has its ones in consecutive
positions, and $\mathbf{A}$ is totally unimodular~\cite{FulkersonGross1965,Schrijver1986}. Consider the polytope of all $z'\in\mathbb{R}^n$
with
\begin{equation}\label{eq:round}
\lfloor \mathbf{A}z^*\rfloor\le \mathbf{A}z'\le\lceil \mathbf{A}z^*\rceil,
\qquad 0\le z'\le1,
\qquad z'_j=z^*_j\ \text{for all forced }j .
\end{equation}
Its constraint matrix consists of $\mathbf{A}$, $-\mathbf{A}$, and rows of the identity matrix and its negative, and is therefore
again totally unimodular. All right-hand sides are integers, being floors and ceilings, the bounds $0$ and $1$, and forced coordinates of $z^*$, which are $0$ or $1$; and
the polytope is nonempty, as it contains $z^*$, and bounded.
Hence every vertex of the polytope is integral, and a vertex $z'\in\{0,1\}^n$ can be computed in polynomial
time by linear programming~\cite{Schrijver1986}. The labeling $z'$ is legal, since its forced coordinates agree with $z^*$, and for every
elementary interval $P$, either $(\mathbf{A}z^*)_P$ is an integer, in which case~\eqref{eq:round} forces $(\mathbf{A}z')_P=(\mathbf{A}z^*)_P$, or
$(\mathbf{A}z')_P$ is its floor or its ceiling, and in both cases
\[
\bigl|(\mathbf{A}z')_P-(\mathbf{A}z^*)_P\bigr|<1
\qquad\text{for every elementary interval }P.
\]
Consequently, for every admissible pair $(P,Q)$,
\[
Z_{z'}(P)+d(Q)-Z_{z'}(Q)\;<\;Z_{z^*}(P)+d(Q)-Z_{z^*}(Q)+2\;\le\;K^*+2\;\le\;\OPT+2 ,
\]
and since the left-hand side is an integer, it is at most $\OPT+1$. By \cref{lem:width}, $W(z')\le\OPT+1$, and
\cref{lem:chains} converts a minimum chain cover of $\prec_{z'}$ into a feasible schedule on at most $\OPT+1$
machines. All steps run in polynomial time, which completes the proof of \cref{thm:intro-approx}.

In summary, the algorithm checks the condition of \cref{obs:feasible}, solves the linear relaxation of~\eqref{eq:ip},
computes a vertex of the polytope~\eqref{eq:round}, and covers $\prec_{z'}$ by chains via bipartite matching. It
consists of two linear programs with $O(n)$ variables, having $O(n^2)$ and $O(n)$ constraints respectively, and one matching computation.

\section{Conclusion}\label{sec:conclusion}


We have shown that requiring a break of length $x$ on every used machine leaves fixed-interval scheduling
polynomial-time solvable for $x\le1$ but turns it into a strongly NP-hard problem for every fixed $x\ge2$.
Despite this hardness, a feasible schedule on at most $\OPT+1$ machines can be computed in polynomial time, and
this is the best guarantee possible unless $\mathrm{P}=\mathrm{NP}$.

Several questions remain. First, the additive-one algorithm solves a linear program with $\Theta(n^2)$
constraints, and it would be interesting to know whether a purely combinatorial algorithm achieves the same
guarantee. Second, one may ask for a characterization of the instances on which the width linear program has an integral
optimum, and for an explicit natural family of instances on which the rounding step incurs the additive loss.
Finally, natural extensions such as breaks restricted to given time windows or machine-dependent costs contain the
present problem as a special case and so inherit the hardness, but not obviously the additive-one guarantee.
Moreover, the normal form of \cref{sec:normal-form} breaks down as soon as a machine may rest twice.

\section*{Acknowledgements}
\aidisclosure

\clearpage
\bibliographystyle{plainurl}
\bibliography{refs}

\end{document}